%% file: _arxiv.tex
\documentclass[a4paper,UKenglish,cleveref, autoref, thm-restate]{lipics-v2021}

\input{lipicsmacros}

\newif\ifanonymous

\hideLIPIcs  

\title{Brief announcement: A special case of maximum flow over time with network changes}

\titlerunning{A special case of maximum flow over time with network changes} 
\ifanonymous
\author{Anonymous}{anonymous institution}{}{}{}

\authorrunning{Anonymous}

\Copyright{Anonymous}

\else 

\author{ Shuchi {Chawla} }{Department of Computer Science, University of Texas, Austin, TX, USA  \and \url{cs.utexas.edu/~shuchi} }{shuchi@cs.utexas.edu}{https://orcid.org/0000-0001-5583-2320}{}

\author{ Kristin {Sheridan}\footnote{A portion of this research performed while at NASA's Glenn Research Center, Cleveland, OH, USA } }{Department of Computer Science, University of Texas, Austin, TX, USA  \and \url{cs.utexas.edu/~kristin} }{kristin@cs.utexas.edu}{https://orcid.org/0000-0002-5524-3259}{This research was funded in part by JPMorgan Chase \& Co.  Any views or opinions expressed herein are solely those of the authors listed, and may differ from the views and opinions expressed by JPMorgan Chase \& Co. or its affiliates. This material is not a product of the Research Department of J.P. Morgan Securities LLC. This material should not be construed as an individual recommendation for any particular client and is not intended as a recommendation of particular securities, financial instruments or strategies for a particular client.  This material does not constitute a solicitation or offer in any jurisdiction.}

\authorrunning{S. Chawla and K. Sheridan} 

\Copyright{Shuchi Chawla and Kristin Sheridan} 

\fi

\ccsdesc[100]{Theory of computation~Design and analysis of algorithms~Mathematical optimization~Discrete optimization~Network optimization} 

\keywords{maximum flow, dynamic flows, flows over time} 

\category{} 

\relatedversion{} 

\funding{This research was funded in part by NSF awards CCF-2217069 and CCF-2225259.}

\acknowledgements{We want to thank the reviewers for their suggestions that helped us improve this paper.}

\nolinenumbers 

\EventEditors{John Q. Open and Joan R. Access}
\EventNoEds{2}
\EventLongTitle{42nd Conference on Very Important Topics (CVIT 2016)}
\EventShortTitle{CVIT 2016}
\EventAcronym{CVIT}
\EventYear{2016}
\EventDate{December 24--27, 2016}
\EventLocation{Little Whinging, United Kingdom}
\EventLogo{}
\SeriesVolume{42}
\ArticleNo{23}

\begin{document}

\maketitle

\begin{abstract}
We consider the problem of finding the value of a maximum flow over time in a network with uniform edge lengths where the edge capacities change at specific time instants.    To solve this problem, we show how to construct a condensed version of a Time Expanded Network (cTEN) whose standard max flow value is the same as the max flow over time on the original network. In particular, for a graph with $n$ nodes, $m$ edges, and $\mu$ {\em critical times} where some edge capacity changes, we obtain a cTEN with $O(n^2\mu)$ nodes and $O(\mu mn)$ edges. This implies that the problem can be solved in $O(\mu^2n^3m)$ time using the combinatorial max flow algorithm of Orlin \cite{orlin2013}, or in $O(\mu^{(1+o(1))}(nm)^{1+o(1)}\log (UT))$ time using the algorithm of Chen et al. \cite{chen2022}, where $U$ is the maximum capacity of any edge and $T$ is the time horizon. We focus on graphs that experience many time changes across the period of interest, as in such graphs the $\mu$ term dominates the runtime.
\end{abstract}

\input{intro}

\input{def}

\input{proofs}

\input{appendix}

\bibliography{main.bib}

\appendix

\input{other_times}

\end{document}

%% file: lipicsmacros.tex
\usepackage[dvipsnames]{xcolor}

\newcommand{\change}[1]{{\color{Mahogany} #1}}

\usepackage{algorithm}
\usepackage[noend]{algpseudocode}

\renewcommand\tilde{\widetilde}
\usepackage{amsopn}

\renewcommand{\st}{^*}
\newcommand{\set}[1]{\{#1\}}
\newcommand{\reals}{\mathbb{R}}

\newcommand{\cT}{\mathcal{T}}

\newcommand{\moveup}[1]{\phi_{#1}^+}
\newcommand{\movedown}[1]{\phi_{#1}^-}

\newcommand{\capac}{u}
\newcommand{\breaks}{\mathcal{T}}
\newcommand{\integers}{\mathbb{Z}}
\newcommand{\ten}{\textsc{TEN}}
\newcommand{\cten}{\textsc{cTEN}}
\newcommand{\critcut}[1]{\textsc{Range}(#1)}

\newcommand{\critnet}[1]{\textsc{CriticalTimes}(#1)}

\newcommand{\forbidden}[3]{\texttt{forbidden}_{#1, #2}(#3)}

%% file: intro.tex
\section{Introduction}



With the rapid ongoing deployment of constellations of small and nano satellites, intersatellite communication systems such as Starlink are quickly realizing the potential of a super fast ``space internet'', bringing access to remote parts of the world \cite{Handley}. But these systems present new challenges for algorithm design: standard routing and communication protocols designed for static networks do not work as-is on space networks.
Objects in space are constantly in motion, and the ability for one object to communicate with another may exist at some times but not at others, such as when their connection is blocked by a planetary body or other object (see Figure~\ref{fig:space-network}). Further, the movement of objects may change the time it takes a message to travel from one point to another, depending on when the message departs. This adds temporal effects to routing in space networking that are not present in traditional networking. 

Temporal effects are also present in terrestrial networks, such as  transportation networks. Consider a shipping company like FedEx transporting large loads between different cities. The scheduling of vehicles faces many temporal constraints -- coordination with scheduled flights;
the availability of vehicles or drivers varying over time, etc. These constraints place transport networks outside the realm of settings most routing or flow algorithms are designed for.

\begin{figure}
    \centering
    \includegraphics[width=\linewidth]{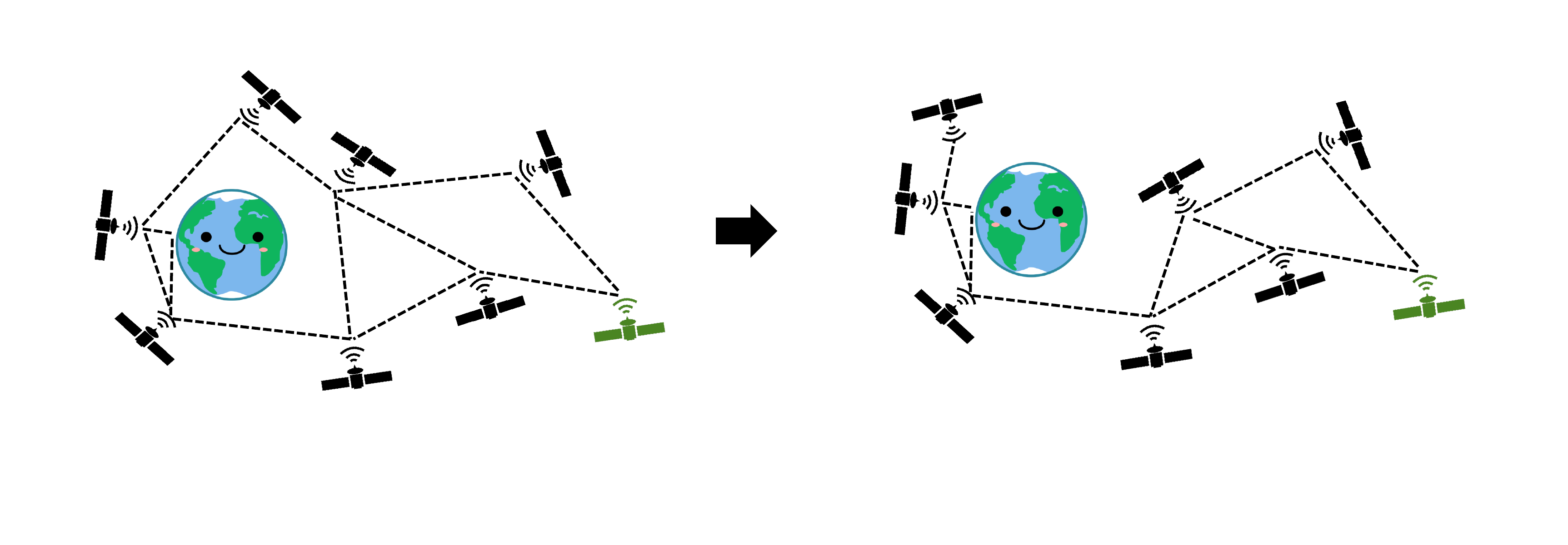}
    \caption{A visualization of the types of changes that can happen in space networks. Due to changes in orientation or movement of other celestial objects, capacities and travel times associated with individual connections may change over time. }
    \label{fig:space-network}
\end{figure}

We focus on a network model where the capacities of edges can vary with time but those variations are known in advance. This is a reasonable model for both of the applications mentioned above. In space networking, for example, objects move according to predictable patterns. The standard ``contact graph'' model \cite{fraire2021,hylton2022} assumes that connections between network nodes exist only for a specified period of time, but all such periods are known ahead of time, and each connection, or edge, has an associated length and capacity. Likewise, for transportation networks, transit times and other temporal constraints are often fixed in advance or predictable. We call network parameters that vary over time {\bf temporal} and parameters that are constant {\bf static}. Temporal networks are networks that contain some temporal parameters, while static networks do not.
{Note that flow problems on static networks \cite{Skutella2009} are different from {\em steady-state} or {\em classical} flow problems (i.e. the traditional flow networks introduced by Ford and Fulkerson \cite{ford1956}), as in the former, flow takes time to traverse edges.}

Our primary question is: {\bf Given a temporal network $N$, a source $s$, a sink $d$, and a time horizon $T$ what is the maximum amount of flow that can be moved from $s$ to $d$ over time horizon $T$?} We call this the maximum flow over time problem on temporal networks, and in this paper we restrict our attention to a special type of network where all edges have a single static length $\tau$,\footnote{In Appendix \ref{sec:otherlengths}, we discuss an extension of our algorithm that is parameterized by (and in fact exponential in) the number of distinct edge lengths in the network.} and all edge capacity functions are piecewise constant.

Let us now specify some key parameters and features of our setting. We consider graphs on $n$ nodes and $m< n^2$ possible directed edges. All edges have a fixed length $\tau\geq 0$, and each edge $ij$ has piecewise constant capacity function $u_{ij}$. We define $\mu_{ij}$ to be the number of constant pieces of the function $u_{ij}$, and $\mu:=\sum_{ij\in E}\mu_{ij}$. That is, $\mu$ is (up to a constant factor) the total number of times {\em any} edge undergoes a capacity change.
We focus on graphs where $\mu$ is very large compared to the number of nodes or edges in the graph (i.e. where the time horizon is very long and the network undergoes many changes), so our primary goal will be to minimize dependence on $\mu$ in our final runtime. Further, note that networks of this class can be specified in space $\Theta(\mu\log U+\log T)$ where $U$ is the maximum capacity of any edge at any time, whereas a network with fully general capacity functions could take $\Omega(T)$ space to describe, as the capacity at each time step must be specified. We assume in this paper that our input networks are always concisely specified using $\Theta(\mu\log U+\log T)$ space.


Most solutions to max flow problems on temporal graphs go through reduction to the transshipment feasibility problem. In the transshipment feasibility problem (TFP), we are given a (temporal or static) network on $n$ nodes and $m$ edges, as well as a set of supplies $v_i$ for each node $i$ in the graph such that $\sum_{i\in V}v_i=0$ and a time horizon $T$. The transshipment feasibility problem asks if there is a flow over the time interval $[0,T]$ such that the net flow out of each node $i$ is exactly $v_i$. Hoppe and Tardos \cite{hoppe2000} give a polynomial time reduction from the transshipment feasibility problem on {\em temporal} networks with piecewise constant edge travel times and capacities to the transshipment feasibility problem on {\em static} networks. 
That is, given a temporal network and a set of supply values for TFP  with $n$ nodes and $m$ edges and where edge $ij$ has $\mu_{ij}$ pieces in its piecewise constant capacity and length functions, Hoppe and Tardos \cite{hoppe2000} construct a static network on $n+4\mu$ nodes and $5\mu$ edges, where $\mu:=\sum_{ij\in E}\mu_{ij}$. If the original temporal network had $k$ terminals (nodes with non-zero supply), the new static network has $k+2\mu$ terminals. Hoppe and Tardos \cite{hoppe2000} show that TFP can be solved on a static network of $n_s$ nodes, $m_s$ edges, and $k_s$ terminals by finding the minimum value of a submodular function on the terminals. The evaluation oracle for this submodular function is an iteration of minimum cost flow on a steady-state network with $n_s$ nodes and $m_s$ edges. Table \ref{tab:sfm} shows runtimes that are possible for the static TFP using different submodular function minimization algorithms.

\input{table_sfm}


The best known combinatorial algorithm for steady-state min cost flow on networks with $n_s$ nodes and $m_s$ edges is due to Orlin \cite{orlin88} and runs in time $\tilde O(m_s^2)$\footnote{We use $\tilde O$ to suppress terms that are polylogarithmic in $n$ or $\mu$  in big $O$ notation.} on a steady-state network with $m_s$ edges, and the fastest weakly polynomial time algorithm, due to Chen et al. \cite{chen2022}, runs in $ O(m_s^{1+o(1)}\log U)$ time on such a network, where $U$ is the maximum  capacity of any edge in the network.
Using our assumption that in our temporal networks, the number of changes to the network is much larger than the number of edges or nodes (i.e. $\mu>>m,n$) and the temporal-to-static reduction of Hoppe and Tardos \cite{hoppe1995}, this implies that the transshipment feasibility problem on temporal networks can be solved in strongly polynomial time $\tilde O(\mu^5)$ or weakly polynomial time $ O(\mu^{3+o(1)}\cdot \log^{O(1)}(TU))$. 

Notably, the transshipment feasibility problem takes as input a set of supplies $v_i$ on the nodes in the network. Thus, if we want to reduce the max flow over time problem to the transshipment feasibility problem, we can set the source and sink flow values to sum to $0$ and search for the maximum value that leads to a feasible problem. This can be done by pairing a TFP solution with binary search or parametric search.
Parametric search \cite{megiddo1978,van2002} roughly squares the runtime of an algorithm that uses only comparisons, additions, and multiplications, and binary search (or exponential search) will increase the runtime by approximately a factor of $\log F\st$, where $F\st$ is the true max flow value. Thus, using the temporal-to-static reduction of Hoppe and Tardos \cite{hoppe2000}, the maximum flow value for a dynamic network over a time horizon $T$ can be found in time $ \tilde O(\mu^{3+o(1)}\log^{O(1)}(UT))$ or in strongly polynomial time $\tilde O(\mu^{10})$.\footnote{There may be better ways to optimize search for the optimal maximum flow  value with these specific algorithms, but a naive implementation and analysis of binary or parametric search with the algorithms mentioned above leads to these runtimes. }

The algorithms we have  discussed for static TFP work for {\em any} network that has piecewise constant edge capacities and travel times. If we restrict the input to have lengths that are static (not changing with time) and uniform (the same) across all edges, the min cost flow algorithms used as a subroutine in many of these algorithms may get faster, as the costs become uniform. However, the temporal-to-static reduction of Hoppe and Tardos \cite{hoppe2000} does {\em not} preserve the uniformity of edge lengths and in fact creates new edges with lengths that depend on the times at which edge capacities change.\footnote{Despite the fact that edge lengths are not preserved, the reduction of Hoppe and Tardos \cite{hoppe2000} does create a network with a very specific structure that may be exploitable in other ways.} 

In this paper, we consider the special case of the temporal maximum flow problem in which all edges have uniform, static length. We show that this special case can be solved in time $O(MF(\Theta(n^2\mu),\Theta(mn\mu)))$, where $MF(n',m')$ is the runtime of a steady-state maximum flow algorithm on $n'$ nodes and $m'$ edges. Using the $O(m'n')$ max flow algorithm of Orlin \cite{orlin2013}, we can thus solve this problem in strongly polynomial $O(\mu^2n^3m)$ time, or using the  algorithm of Chen et al. \cite{chen2022}, we can solve the problem  in time $ O(\mu^{1+o(1)}(nm)^{1+o(1)}\cdot \log (nTU))$.\footnote{ Although our algorithm only finds the {\em value} of the temporal maximum flow, and not a particular flow, the temporal-to-static TFP reduction of Hoppe and Tardos \cite{hoppe2000} paired with the quickest transshipment solution of Anapolska et al. \cite{anapolska2025} can be used to find a feasible flow of a specified value (if it exists) in time $\tilde O(\mu^7)$. Our runtimes and the ones we compare them to are specifically for finding the {\em value} of the maximum flow. } Table \ref{tab:maxflow} compares runtimes of different algorithms for this problem.

In the regime where $\mu>>n,m$ (i.e. the graph undergoes many changes relative to the number of edges and nodes), we get a better runtime for the maximum flow problem in both the strongly and weakly polynomial settings. In particular, if $m,n\leq\log(\mu)$, we get runtimes of $ O(\mu^2)$ and $ O(\mu^{1+o(1)}\log (nTU))$, as compared to existing runtimes of $\tilde O(\mu^{10})$ or $ O(\mu^{3+o(1)}\log^{O(1)}(nTU))$. 
Even in a less extreme regime, such as where $n,m=O(\mu^{1/4})$, we get runtimes $O(\mu^3)$ and $ O(\mu^{3/2+o(1)}\log(nTU))$. 

\input{tab_maxflowval}

\paragraph*{Overview of techniques}

Existing methods for finding the maximum flow over time with temporal edge capacities either apply to the setting in which all edge lengths are $0$ \cite{ogier1988,fleischer1999} or go through the Hoppe-Tardos reduction to the static transshipment feasibility problem \cite{hoppe2000}. The latter approach works for general edge lengths, including piecewise constant edge lengths. Current techniques for the static transshipment feasibility problem define a function $o(A)$ for any $A$ that is a subset of terminals in the network such that $o(A)$ is the maximum amount of flow that can be sent from the sources (nodes with positive supply) in $A$ to the sinks (nodes with negative supply) in $A$ over time horizon $T$. The transshipment feasibility problem is feasible on a network if and only if for all subsets $A$ of the terminals, $o(A)-v(A)\geq 0$, where $v(A):=\sum_{i\in A}v_i$. The smallest value of $o(A)-v(A)$  can be  found using submodular function minimization.

For a fixed set $A$, the function $o(A)$ can be computed by finding the maximum flow over time on a static network with a single super-source and a single super-sink. The maximum flow over time problem for static networks can be reduced to the classical steady-state minimum cost flow problem, so
min cost flow algorithms for steady-state networks can be used as an evaluation oracle for $o(A)-v(A)$. Note however that if we use the temporal-to-static reduction of Hoppe and Tardos and then reduce the static flow over time problem to steady-state min cost flow, we ultimately work on a network of $\Theta(\mu)$ terminals, nodes, and edges, so the submodular function minimization and min cost flow runtimes will depend on $\mu$. See Figure \ref{fig:htreduction} for a visualization of the reduction process defined by Hoppe and Tardos \cite{hoppe2000}.\footnote{Note also that the reductions in this figure are specifically for the {\em feasibility} of the transshipment problem. To find an actual transshipment meeting the given parameters, Hoppe and Tardos \cite{hoppe2000} use an additional reduction that goes through another problem that they refer to as lexicographical maximum flow.} 

\begin{figure}
    \centering
    \includegraphics[width=\linewidth]{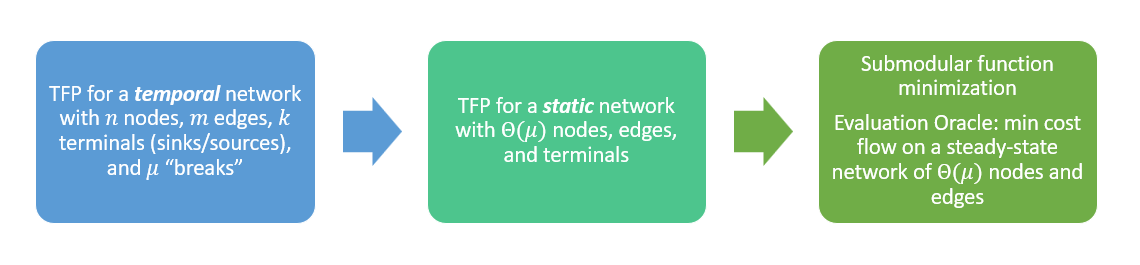}
    \caption{Hoppe and Tardos \cite{hoppe2000} gave a reduction from  TFP on temporal networks with $n$ nodes, $m$ edges, and $k$ terminals to the TFP on static networks with $\Theta(\mu)$ nodes, edges, and terminals, where $\mu:=\sum_{ij\in E}\mu_{ij}$ and $\mu_{ij}$ is the number of pieces in the piecewise constant capacity function for edge $ij$. They also show that the TFP can be solved for a static network on $n'$ nodes and $m'$ edges using submodular function minimization where the evaluation oracle is a steady-state/classical min cost flow on a network with $O(n')$ nodes and $O(m')$ edges. }
    \label{fig:htreduction}
\end{figure}

Our main novel technical contribution is in defining what we call the condensed Time Extended Network (cTEN) by identifying a set of ``critical points'' in the full time expanded network (TEN). The TEN for a network with time horizon $T$ is a static network on node set $V\times[0,T]$\footnote{In this paper, we use $[a,b]$ for $a\leq b$ to refer to the set $\set{a,a+1,a+2,\ldots,b-1,b}$.} whose maximum flow value is equal to the maximum flow over time $T$ on the original network.  We show that if the original network has uniform edge lengths, there is a small set of times $\cT$ such that the TEN must contain a minimum $(s,0)$-$(d,T)$ cut in which if $(i,t)$ is on $(d,T)$'s side of the cut and $(i,t+1)$ is on $(s,0)$'s side of the cut for any $i\in V$, then $t\in \cT$. By the structure of the TEN, we also have that if the min cut has finite capacity, $(i,t)$ being on $(s,0)$'s side of the cut implies $(i,t+1)$ is also on $(s,0)$'s side of the cut, so this set $\cT$ fully captures the times where a node can ``switch sides'' of the cut. The existence of such a set lets us build the cTEN, which is the same as the standard TEN, except that it groups the nodes of $V\times [0,T]$ into intervals whose start times are in $\cT$, reducing the node set to $V\times \cT$. We show that if $\cT$ is chosen appropriately, the minimum $(s,0)$-$(d,T)$ cut value on the cTEN will be the same as the minimum cut value on the full TEN. Thus, the maximum flow value will also be the same. Then the max flow value can be found using a single steady-state maximum flow call on a network of $O(n^2\mu)$ nodes and $O(\mu nm)$ edges, rather than the multiple calls to minimum cost flow algorithms on networks of $\Theta(\mu)$ nodes and edges used by existing techniques. Further, our single call can directly find the maximum flow  over time value, rather than checking if a given flow value is feasible, so it does not require an extra layer of parametric or binary search to find the optimum. A drawback is that it only works on networks with uniform edge lengths.\footnote{Technically the algorithm works for networks whose edge lengths are all static and either $0$ or $\tau$ for some fixed value $\tau$.}

\subsection*{Related work}


A wide variety of problems have been studied in the context of  both steady-state flows and flows over time, and many works focus on different regimes -- such as different parameters that should be minimized or maximized, the number of sources and sinks in the network, and the temporal properties of the graph parameters. We briefly outline some of the most relevant variants of this problem and associated results here. This section is optional and not necessary to understanding our proofs, which only require a familiarity with steady-state (classical) max flow and min cut theorems, but we hope to outline where our results fit in the current landscape.

\begin{itemize}
    \item {\bf Maximum steady-state flows.}
    The notion of steady-state flows was first introduced by Ford and Fulkerson \cite{ford1956}. In the maximum steady-state flow problem we are given a network on $n$ nodes and $m$ directed edges, each with a specified capacity that upper bounds the flow that can be moved over that edge. We are asked for the maximum amount of flow we can move from the source $s$ to the  sink $d$ such that the net flow into each node in the network other than $s$ or $d$ is $0$ and all edge capacity constraints are respected.\footnote{Note that in a flow over time problem with an infinite time horizon, the edge lengths do not matter and we simply look for a constant rate of flow that maximizes the total flow into $d$, which is the same as the classical max flow problem and why we refer to it as the steady-state problem.}

    In a crucial use of strong duality, Ford and Fulkerson \cite{ford1956} showed that the value of a maximum steady-state $s$-$d$ flow in a network is equal to the value of the minimum $s$-$d$ cut in that network. An $s$-$d$ cut in a network with node set $V$ is a set $S\subseteq V$ such that $s\in S,d\notin S$, and the cost of such a cut is the sum of the capacities of all edges that begin in $S$ and end in $V\setminus S$.
    
    Many authors have found strongly polynomial \cite{edmondskarp,blockingflows,pushrelabel,orlin2013} as well as pseudopolynomial  and weakly polynomial solutions \cite{ford1956,goldbergrao1998,chen2022} for maximum steady-state flow. In particular, Orlin \cite{orlin2013} showed that the problem on a network of $m$ edges and $n$ nodes can be found in time $O(nm)$, and Chen et al. \cite{chen2022} gave an algorithm that finds a (fractional) steady-state maximum flow in time $O(m^{1+o(1)}\log U)$, where $U$ is the maximum edge capacity. 

    \item {\bf  Minimum  cost  steady-state flows and circulations.} In the  minimum cost steady-state flow problem, we are again given a network on $n$ nodes and $m$ edges, but now each edge of the network is associated with some cost per unit of flow as well as a capacity bound. We are also given a source $s$, sink $d$ and a target flow value $F\st$. The problem asks for a flow through the network that moves $F\st$ total material from $s$ to $d$ and has minimum cost among all such flows. The minimum cost circulation problem is similar, except that the net in-flow amount at {\em every} node must be $0$. The two problems are largely interchangeable via small network changes.
    
    The minimum cost steady-state flow problem also dates back to Ford and Fulkerson \cite{ford1958}, and many authors have found efficient algorithms for the minimum cost flow algorithm \cite{goldbergtarjan1989,goldbergtarjan1990,goldbergrao1998,orlin1997,orlin88,chen2022}. In particular, Orlin \cite{orlin88} gave an algorithm that works in $\tilde O(m^2)$ time, and Chen et al. \cite{chen2022} gave one that works in $O(m^{1+o(1)}\log U)$ time, where $m$ is the number of edges in the network and $U$ is the maximum edge capacity.

    \item {\bf Maximum flows over time on static networks.} 
    In the maximum flow over time problem for static networks, we are given a network on $n$ nodes and $m$ edges, as well as a constant flow capacity per time step and length for each edge. Given a source $s$, sink $d$, and time horizon $T$, the goal is to find a flow over each edge across time steps $0$ through $T$ such that (1) edge capacities are not violated at any time step, (2) the net  flow across the entire time period into every node other than $s$ and $d$ is $0$ and it is never negative at any point in time, and (3) the net flow into $d$ across the entire time period is maximized. Some variations of flow over time problems further restrict the ``storage'' at each node; that is, they provide extra parameters that bound the net amount of in-flow that each node can have at {\em any} time in the time period. 

    Ford and Fulkerson \cite{ford1958} first introduced maximum flows over time, and they showed that for static networks the problem can be reduced to the minimum cost circulation problem. However, there are many variations of maximum flow over time for static networks. Skutella 
    \cite{Skutella2009} provides a nice summary of these variations and some existing results, and we briefly outline some of these problems here. It turns out that for static networks, storage capacity at non-sink nodes is not needed, so existing optimality results are for any choice of bounds on storage capacity.
    
    \begin{itemize}
        \item {\it Quickest flow.} In the quickest flow problem, we are given a desired flow value $F\st$ and asked for the shortest time horizon $T$ such that there exists a flow over time $T$ of value $F\st$ from $s$ to $d$. The continuous and discrete time versions of this problem have been studied by many authors including \cite{burkard1993,fleischer2002,fleischer2007,saho2017}.

        \item {\it Universal maximum flow/earliest arrival flow.} In the universal or earliest arrival flow problem, we are given a time horizon $T$ and asked for a flow over time horizon $T$ such that for all $t\leq T$ the amount of flow that has been delivered from $s$ to $d$ by time $t$ is the maximum flow possible over time horizon $t$. Perhaps surprisingly, several authors have found that such a flow always exists for static networks, and they have given algorithms to find such a flow \cite{Tjandra2003,baumann2009earliest,schmidt2014}.

    \end{itemize}

    \item {\bf Transshipments  on static networks.} 
    In the feasibility variant of the transshipment problem for static networks, we are given a static network with $n$ nodes and $m$ edges, a time horizon $T$, and a supply at every node in the network. The problem asks if there is a flow over time horizon $T$ that results in a net out-flow amount at each node exactly equal to its supply. Many authors have studied this problem \cite{hoppe1995,hoppe2000,schloter2017,schloter2022,skutella2023,anapolska2025}. 
    Notably, Hoppe and Tardos  \cite{hoppe2000} gave the first algorithm to decide if a given transshipment is feasible, and they showed how to find such an integral flow, if it exists, using a feasibility oracle for the problem. Later, Anapolska et al. \cite{anapolska2025} gave a more efficient algorithm for finding an integral transshipment. 
    
    Another popular transshipment problem is that of the {\em quickest transshipment}. In this problem, we are given a set of demands and asked for the {\em smallest} time horizon $T$ such that there is a feasible transshipment over time horizon $T$ that meets the given set of demands. Note that once the optimal time horizon is known, an algorithm that takes in a fixed time horizon and finds a feasible transshipment is sufficient for finding a quickest transshipment. Hoppe and Tardos  \cite{hoppe2000} showed how to find the optimal time horizon for the quickest transshipment problem on static networks, and Schl\"oter et al. 
    \cite{schloter2022} later gave an algorithm with improved runtime. 

    These results assume unbounded storage at each node in the network.


    \item {\bf Extensions to temporal networks.} 
    Transshipment and flow over time problems can be extended to the setting of temporal networks, or networks whose parameters, such as edge lengths, edge capacities, and node storage capacities can vary over time. It turns out that storage is essential for polynomial time solutions for even the simplest variant of this problem. 

    In particular, Zeitz \cite{zeitz2023} showed  via reduction to the subset-sum problem that given a network with temporal edge capacities (or edge lengths) and no storage capacity at any node as well as a time horizon $T$, it is NP-hard to decide if it is possible to move one unit of flow from $s$ to $d$ over time horizon $T$. In fact, it is hard to decide if $0$ or $F$ flow is achievable, for any choice of $F>0$, so the maximum flow over a given time horizon is not even approximable if there is no storage capacity at any node. Note that the temporal-to-static reduction of Hoppe and Tardos for the quickest transshipment problem means this hardness result also applies to transshipments on static networks without storage.

    Other authors have studied a variety of flow problems on temporal networks with unbounded storage capacity at every node.

    \begin{itemize}
        \item {\it Maximum flow over time.} Akrida et al.
        \cite{akrida2019}  
        gave an algorithm for the maximum flow over time problem on temporal networks with bounded node storage that takes time polynomial in the number of time steps at which some edge has non-zero capacity and for which all edge lengths are $0$. They also use a type of condensed time extended network, but their approach is quite different from ours, and their algorithm is not polynomial in the input size when the input includes arbitrary but concisely specified piecewise constant capacities. 
        
    
        \item {\it Universal maximum flow/earliest arrival flow.} Ogier
        \cite{ogier1988} and Fleischer \cite{fleischer1999} gave algorithms to find earliest arrival flows for networks with piecewise constant capacities in which all travel times are $0$. Our approach here is similar in spirit to the approach of  Ogier \cite{ogier1988}, though we deal with uniform (but non-zero) travel times and we produce the value maximum flow and do not consider an earliest arrival flows. 

        \item {\it Transshipments.}
         In their seminal work, Hoppe and Tardos \cite{hoppe2000} gave a reduction for the transshipment feasibility problem on temporal networks with piecewise constant edge capacities to the same problem on static networks. This reduction takes a temporal network on $n$ nodes and $m$ edges such that $\mu:=\sum_{ij\in E}\mu_{ij}$, where $\mu_{ij}$ is the number of constant pieces in the capacity function for edge $ij$ and produces a static network on $\Theta(\mu)$ nodes and edges. Notably, this reduction generates demand on $\Theta(\mu)$ nodes, and many of the edges have lengths that depend on the breaktimes of the piecewise constant capacity functions in the input. 
    \end{itemize}


\end{itemize}

%% file: table_sfm.tex
\begin{table}[]
    \centering
    \begin{tabular}{c|c}
       Runtime  & SFM algorithm  \\
       \hline 
       $O(k_s^3 \log k_s\cdot MCF(n_s,m_s)+k_s^4\log k_s)$  & Jiang, Lee, Song, and Zhang \cite{jiang2024convex} \\
       $O(k_s^2\log(nUT)\cdot MCF(n_s,m_s) + k_s^3\cdot \log^{O(1)}(nUT))$ & Lee, Sidford, and Wong \cite{lee2015sfm}
    \end{tabular}
    \caption{This table shows the runtime of TFP for a time horizon $T$ on static networks with $n_s$ nodes, $m_s$ edges, $k_s$ terminals, and $U$ maximum edge capacity using different submodular function minimization (SFM) algorithms. In this table, $MCF(n_s,m_s)$ is the runtime for minimum cost flow on steady-state networks of $n_s$ nodes and $m_s$ edges. }
    \label{tab:sfm}
\end{table}

%% file: tab_maxflowval.tex
\begin{table}[]
    \centering
    \begin{tabular}{c|c|c}
        Algorithms used & runtime & network type  \\
        \hline 
         \cite{hoppe2000}, \cite{jiang2024convex} & $O(\mu^3\log \mu \cdot MCF(\mu,\mu) + \mu^4 \log \mu)$ & all networks$\st$ \\
         \cite{hoppe2000}, \cite{lee2015sfm} & $O(\mu^2\log(nUT)\cdot MCF(\mu,\mu) + \mu^3 \cdot \log^{O(1)}(nUT))$ & all networks$\st$ \\
         \cite{hoppe2000}, \cite{orlin1997} \cite{jiang2024convex} & $O(\mu^5\log^2 \mu) $  & all networks$\st$ \\
         \cite{hoppe2000}, \cite{chen2022},\cite{lee2015sfm} & $O(\mu^{3+o(1)}\log^2(nUT) + \mu^3\log^{O(1)}(nUT))$ & all networks$\st$ \\
         this paper & $O(MF(n^2\mu,mn\mu))$ & uniform edge lengths \\
         this paper, \cite{orlin2013} & $O(\mu^2\cdot ( n^3m))$ & uniform edge lengths \\
         this paper, \cite{chen2022} & $O((\mu nm)^{1+o(1)}\log(nUT))$ & uniform edge lengths
    \end{tabular}
    \caption{This table compares the time to find the maximum flow over time $T$ for a temporal network using different techniques. In particular, it presents the time to find such a maximum flow value (not necessarily a specific flow) for a temporal network with $n$ nodes, $m$ edges, and maximum edge capacity $U$ such that the edge capacity functions are piecewise constant, with at most $\mu$ total constant pieces across all edge capacity functions. $MCF(n',m')$ is the runtime of a steady-state minimum cost flow algorithm on a network of $n'$ nodes and $m'$ edges, and $MF(n',m')$ is the runtime of a steady-state maximum flow algorithm on a network of the same size. \\    
    $\st$These algorithms verify if a given flow value is possible. To find the largest possible flow value, they can be paired with a search method such as parametric or binary search. }
    \label{tab:maxflow}
\end{table}

%% file: def.tex
\section{Definitions and main result}


We begin by formally defining a general version of temporal networks. Intuitively, a temporal network is parameterized by its edge capacities and lengths and its node storage capacities. The edge capacities specify how much flow  can traverse that edge in a single time step; the lengths specify how long it takes flow to traverse that edge; and the node storage capacities denote how much flow each node can store between time steps.

\begin{definition}
A time-varying, or temporal, network \\ $N=(V,E,\{\tau_{ij}:\integers\rightarrow \reals_{>0}\}_{ij\in E},\{\capac_{ij}:\integers\rightarrow \reals_{\geq 0}\}_{ij\in E},\set{a_i:\integers\rightarrow \reals_{\geq 0}}_{i\in V\setminus\set{s,d}})$ consists of 
\begin{itemize}
    \item a set of vertices $V$,
    \item a set of directed edges $E\subseteq V\times V$,
    \item length functions $\tau_{ij}$ for each edge,
    \item capacity functions $\capac_{ij}$ for each edge that denote the capacity of edges at each integral time step, and
    \item node storage capacity functions $a_i$ that denote the amount of flow that can be stored at a particular node at each time step.
\end{itemize}
\end{definition}

Further, if $\tau_{ij},u_{ij},$ and $a_i$ are constant for all $i,j\in V$, we refer to the network as {\em static}. In this paper, we will assume $a_i(t)=\infty$ for all $i,t$ and as a result we generally exclude $a_i$ from the definition of a temporal network, but first we present the definition of flows over time in full generality.
We define $[0,T]:=\set{0,1,\ldots,T}$.






\begin{definition}
    An $s$-$d$ {\em flow over time $T$} in a network $(V,E,\set{\tau_{ij}},\set{\capac_{ij}},\set{a_i})$ is a set of functions $f_{ij}:[0,T]\rightarrow \reals_{\geq 0}$ for each $ij\in E$ such that the following hold:
    \begin{enumerate}
        \item For all integers $t\in [0,T],ij\in E$, $f_{ij}(t)\leq u_{ij}(t)$ \change{(edge capacity constraint)}

        
        \item For all $t\in [0,T],i\in V\setminus\{s\}$,  \change{(net flow into a node is non-negative at all times)} 

        \begin{align*}
            0\leq \sum_{j:ji\in E}\sum_{t'=0}^{t-\tau_{ji}} f_{ji}(t')  -\sum_{j:ij\in E}\sum_{t'=0}^t f_{ij}(t') 
        \end{align*}

        Further, if the $a_i$ are not uniformly infinite, then for $i\in V\setminus \set{s,d},t\in [0,T]$, \change{(net flow into each node does not exceed storage capacity at that node for that time)}
        \begin{align*}
            \sum_{j:ji\in E}\sum_{t'=0}^{t-\tau_{ji}} f_{ji}(t')  -\sum_{j:ij\in E}\sum_{t'=0}^t f_{ij}(t') \leq a_i(t),
        \end{align*}
       
        \item For all $i\in V\setminus\{s,d\}$, \change{(net flow over the whole period is $0$ for all but $s$ and $d$)}        
            \begin{align*}
                 \sum_{ji\in E}\sum_{t=0}^{T-\tau_{ji}} f_{ji}(t)  -\sum_{ij\in E}\sum_{t=0}^T f_{ij}(t)=0,
            \end{align*}

    \end{enumerate}

    Further, the {\em value} of flow $f$ at over time horizon $T$ is the net flow into $d$ at time $T$, which is
         \begin{align*}
            |f|_{T}:=\sum_{j:jd\in E}\sum_{t=0}^{T-\tau_{jd}} f_{jd}(t)  -\sum_{j:dj\in E}\sum_{t=0}^T f_{dj}(t).
        \end{align*}

    When $T$ is clear from context, we may write $|f|$ instead of $|f|_T$.
\end{definition}

This is a generalization of the original notion of flow introduced by Ford and Fulkerson \cite{ford1956}, and in fact a similar model was introduced by the same authors not much later \cite{ford1958}. The ``standard'' notion of maximum flow, which we refer to as steady-state max flow, is equivalent to  max flow over time in which all travel times are $0$ and $T=0$. (Note that in this setting, we can just define $f_{ij}$ at a single time step $t=0$, and constraint (2) is redundant, so we only need constraints (1) and (3).) 

\textbf{Further, in this paper we make a significant assumption about the structure of networks we consider.} In particular, we assume that there exists some universal constant $\tau$ such that $\tau_{ij}$ is constant and equal to $\tau$ for all $ij\in E$.
Additionally, we assume each $u_{ij}$ is a piecewise constant function and changes value on at most $\mu_{ij}$ distinct time steps. For a given network with edge set $E$, we denote $\mu:=\sum_{ij\in E}\mu_{ij}$. We assume that network capacity functions are specified by listing the value of that function over each interval during which it is constant, so the input size is at least $\mu$. Additionally, we will focus on networks where $\mu $ is very large compared to $n$ and $m$, so our main goal will be to limit dependence on $\mu$.



\subsection*{Time Extended Networks (TENs)}

Steady-state flows have been extensively studied, so we focus on a translation of flow over time problems on temporal networks to steady-state flow problems. The most basic way to relate these two problems is using a Time Extended Network of the temporal network $N$, which we denote as $\ten(N)$.

\begin{definition}
For a network $N=(V,E,\set{\tau_{ij}},\set{u_{ij}})$, source $s$, sink $d$, and time horizon $T$ we define $\ten(N,T)$ as a classical (steady-state) flow network $G$ with the following characteristics:
\begin{itemize}
    \item The vertex set is $V'=V\times [0,T]$.
    \item There is an edge from $(i,t)$ to $(j,t')$ for $i\neq j$ if and only if $ij\in E$, $t'=t+\tau_{ij}$, and $u_{ij}(t)\neq 0$. Further, the capacity of this edge is $u_{ij}(t)$.
    \item For all $i$ and all $t\in [0,T-1]$, there is an edge from $(i,t)$ to $(i,t+1)$ with infinite capacity.
    \item The source is $(s,0)$ and the sink is $(d,T)$.
\end{itemize}

If $T$ is clear from context, we may just write $\ten(N)$.
\end{definition}

For any $s,d$ and any $T$, the maximum $s$-$d$ flow over time value in $N$ is the same as the maximum steady-state $(s,0)$-$(d,T)$ flow in $\ten(N,T)$ \cite{anderson}. Thus, one solution to this problem is to run any steady-state max flow algorithm on $\ten(N,T)$.
However, the size of $\ten(N,T)$ blows up with the value of $T$, so just writing it down could take time exponential in the size of the description of $N$ and $T$.

To exploit the succinct structure of our networks (that is, the structure of networks with piecewise constant capacities), we define a refined version of $\ten(N,T)$, which we call a {\em condensed time extended network}, or $\cten(N,A,T)$ (where $A\subseteq[0,T]$ is a set to be defined). The idea of $\cten(N,A,T)$
 is that $A$ will partition $[0,T]$ into intervals such that each interval begins at some time in $A$.
 
\begin{definition}
If $N=(V,E,\set{\tau_{ij}},\set{u_{ij}})$ is a temporal network with source $s$ and sink $d$ and $A\subseteq [0,T]$ with $0,T\in A$, let $t_1< t_2< \cdots< t_k$ be the times in $A$. We define $\cten(N,A,T)$ as a classical (steady-state) flow network as follows:
\begin{itemize}
    \item The vertex set is $V':=V\times A$.

    \item For all $x\in V$ and all $1\leq j <k$, there is an edge from $(x,t_j)$ to $(x,t_{j+1})$ with infinite capacity.\footnote{Here we use infinity to mean a special (finitely represented) character that is larger than any finite number in comparison. It also works to replace each infinite capacity edge with an edge of capacity larger than the maximum flow (for example, larger than $nUT$, where $U$ is the maximum edge capacity). }
    
    \item For all $xy\in E$ and all $1\leq i,j <k$, there is an edge from $(x,t_i)$ to $(y,t_{j})$ with the following capacity if it is non-zero:
    \begin{align*}
        \sum_{t:(t_i\leq t<t_{i+1}),(t_j\leq t+\tau_{ij} <t_{j+1})}u_{xy}(t).
    \end{align*}

\end{itemize}
Further, we treat $(s,0)$ as the source and $(d,T)$ as the sink. For brevity, we may write $\cten(N,A)$ if $T$ is clear from context.
\end{definition}

Note that $\cten(N,A,T)$ is just $\ten(N,T)$ with all nodes $(v,t)$ for $t\in[t_i,t_{i+1}-1]$ collapsed to a single node $(v,t_i)$, and in fact $\cten(N,[0,T],T)$ is the same as $\ten(N,T)$.
When $A$ is relatively small, $\cten(N,A)$ can be written much more concisely than $\ten(N)$. In the piecewise constant capacity case, the capacity of new edges can be computed efficiently as well by using multiplication over intervals instead of strictly adding as specified above. 

\subsection*{Finite cuts and cut functions}

Cuts are a fundamental tool of flow analysis that we use here, so next we describe cuts on steady-state networks (i.e. networks in which all edge lengths are $0$ and with time horizon $0$). If $N=(V,E)$ is a steady-state network, an $s$-$d$ cut $C$ of $N$ for $s,d\in V$ is a partition of the nodes of $V$ such that $s$ and $d$ are on opposite sides of the partition. The cost of this cut is the collective capacity of edges that go from $s$'s part of the partition (denoted $C(s)$) to $d$'s part of the partition (denoted $C(d)$). A classic result of Ford and Fulkerson \cite{ford1956} states that the minimum cost $s$-$d$ cut on a steady-state network is the same as the maximum $s$-$d$ flow value. Due to the simple structure of cuts, we turn our attention to bounding the values of minimum cuts. 

In the remainder of the paper, we only consider finite capacity cuts. Note that if a steady-state network has an edge $ij$ of infinite capacity, a finite cut $C$ on that steady-state network cannot have $i\in C(s),j\in C(d)$. Crucially, this means that if $N$ is a temporal network, then any finite cut on $\ten(N,T)$ that includes $(i,t)\in C((s,0))$ must also have $(i,t')\in C((s,0))$ for all $t'\geq t$. We formalize this in the following observation. 

\begin{observation}\label{obs:cutfunctions}
Let $N$ be a temporal network with  vertex set $V$. Then for any $T$ and any finite cut $C$ on $\ten(N,T)$, $(i,t)\in C((s,0))\implies (i,t')\in C((s,0))$ for all $t'\geq t$ and all $i\in V$. The same is true for $\cten(N,A,T)$ where $A$ is any subset of $[0,T]$ with $0,T\in A$.
\end{observation}

We represent finite capacity cuts in a $\ten$ using a {\em cut function} that maps every vertex $i\in V$ to the earliest time $t\in [0,T]$ such that $(i,t)\in C((s,0))$. Observation \ref{obs:cutfunctions} implies that there is a one-one correspondence between finite capacity cuts in the $\ten$ or $\cten$ and such cut functions.

\begin{definition}
    A cut function $\phi$ on $\ten(N,T)$ (or $\cten(N,A,T)$) maps vertices $i\in V$ to the interval $[0,T+1]$ (or to the set $A\cup\set{T+1}$) with $\phi(s)=0$ and $\phi(d)=T+1$, where $s$ is the source and $d$ is the sink. The cut function uniquely represents the $(s,0)$-$(d,T)$ cut $C_\phi$ with $C_\phi((s,0)):=\set{(i,t):i\in V,t\in [0,T],t\geq \phi(i)}$ (or $C_\phi((s,0)):=\set{(i,t):i\in V,t\in A,t\geq \phi(i)}$).
\end{definition}

We let $cost(\phi)$ be the cost of the $(s,0)$-$(d,T)$ cut defined by $\phi$. 

Note that finite capacity cuts in a $\cten$ form a subset of finite capacity cuts in a $\ten$. Correspondingly, cut functions for $\cten(N,A,T)$ have a smaller range, $A\cup\{T+1\}$, and therefore form a subset of cut functions for $\ten(N,T)$. Thus, we get the following observation. 

\begin{observation}
    For any $N, T$, and $A\subseteq [0,T]$ with $0,T\in A$, the capacity of a min $(s,0)$-$(d,T)$ cut in $\ten(N,T)$ is no larger than the capacity of a min $(s,0)$-$(d,T)$ cut in $\cten(N,A,T)$.
\end{observation}
    
Our goal is to argue that there is a {\em small} $\cten$ that captures the cut structure of the $\ten$, and in particular, has the same min cut capacity as the $\ten$. To this end, we focus on identifying the ``critical'' times corresponding to edges crossing a min cut in the $\ten$. For a cut function $\phi$, let $\critcut{\phi}$ denote the range of $\phi$, that is, 
\[\critcut{\phi}:=\{t: \exists i\in V \text{ with } \phi(i)=t\}.\] 
The following lemma then follows from the correspondence between cut functions of a $\cten$ and a $\ten$.

\begin{lemma}
\label{lem:criticaltimes}
    Let $\phi$ be a cut function corresponding to a min $(s,0)$-$(d,T)$ cut in $\ten(N,T)$, and let $A\subseteq [0,T]$ be any set with $\critcut{\phi}\cap [0,T]\subseteq A$. Then, any min $(s,0)$-$(d,T)$ cut in $\cten(N,A,T)$ is also a min $s$-$d$ cut in $\ten(N,T)$. In particular, the capacity of the min cut in $\cten(N,A,T)$ is equal to the capacity of the min cut in $\ten(N,T)$. 
\end{lemma}

Note that it is {\em essential} that $\critcut{\phi}\cap [0,T]\subseteq A$. Consider the network depicted in Figure \ref{fig:basic-ten}. The collapsed nodes in the \cten\ merge nodes representing times between $1$ and $3$, and as a result the min cut value of the condensed network is much larger than that of the full network.

\begin{figure}
    \centering
    \includegraphics[width=.7\textwidth]{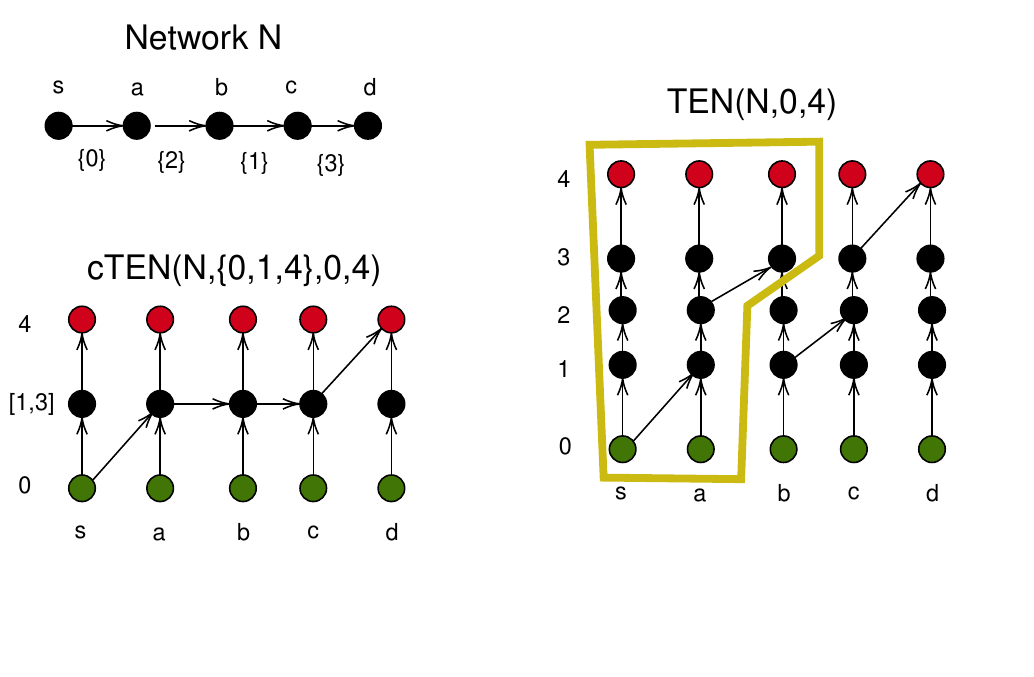}
    \caption{The top left depicts temporal network $N$, which consists of 5 nodes and in which each edge has infinite capacity for one time step (specified below the edge) and has zero capacity on all other time steps. On the right, we have a full time extended network of $N$ for the period $[0,4]$, and on the bottom left we have $\cten(N,\set{0,1,4},4)$. All edges in the networks have infinite capacity. Note that $\ten(N,4)$ has a min cut of capacity $0$ (depicted by the gold box), while $\cten(N,\set{0,1,4},4)$ has no finite $(s,0)$-$(d,4)$ cut.   }
    \label{fig:basic-ten}
\end{figure}

\subsection*{The critical breakpoints of a network}

Having established the basic structure of cuts in a $\ten$, we now identify a small set of times that contains all of the critical times of {\em some} min cut of $\ten(N,T)$. Such a small set $A$ can then be used to construct a small $\cten$ that captures the min cut capacity of the $\ten$ as per Lemma~\ref{lem:criticaltimes}.
Intuitively, we construct our set of critical times to include all times at which some edge changes capacity, as well as certain ``offset'' times from each such change. The idea is that ``information'' about the change in some edge capacity may propagate down a path of nodes in the network, so the offset values should correspond to some possible path length in the network. As we have set $\tau$ to be uniform across all edges, all paths in the network have length $\ell\cdot \tau$ for some $\ell\in [0,n]$. Thus, we define our critical times as follows.



\begin{definition} \label{def:forbidden}
Let $N=(V,E,\tau,\set{u_{ij}})$ be a network with piecewise constant $u_{ij}$, and let $\breaks$ be the set of times at which some $u_{ij}$ changes value. Also include $0$, $T$, and $T+1$ in $\breaks$. Then 
\begin{align*}
    \critnet{N}:= \set{\theta \pm \ell\cdot \tau:\theta\in \breaks,\ell\in [0,n]}.
\end{align*}
\end{definition}

As before, we will let $\mu_{ij}$ be the number of pieces in $u_{ij}$ and $\mu:=\sum_{ij\in E}\mu_{ij}$.
Note that if $\cT$ is as in Definition \ref{def:forbidden}, $|\cT|\leq O(n\mu)$. As an example, consider a network on $n$ nodes such that only one edge ever changes capacity, and it only does so at time $t'$. The critical times of this network would be $\set{0,T}\cup \set{(t'-n\cdot \tau),(t'-(n-1)\cdot \tau ),\ldots,(t'-\tau),t',(t'+\tau),\ldots, (t'+(n-1)\cdot \tau),(t'+n\cdot \tau)}$.  



We can now state our main technical lemma, which we prove in Section \ref{sec:flow-value}:
\begin{lemma}
\label{lem:breakpoints}
    For any $\ten(N,T)$ there exists a cut function $\phi$ corresponding to a min $(s,0)$-$(d,T)$ cut with $\critcut{\phi}\subseteq\critnet{N}$.
\end{lemma}

Note that Lemma \ref{lem:breakpoints} combined with Lemma \ref{lem:criticaltimes} implies that the maximum $s$-$d$ flow over $T$ on $N$ is exactly the steady-state maximum flow value on $\cten(N,A,T)$, where $A=\critnet{N}$. 

Using Lemma \ref{lem:breakpoints} and a short argument to bound the size of $\cten(N,A,T)$, we get our main theorem.

\begin{theorem}
If $N=(V,E,\tau,\set{u_{ij}})$ is a temporal network with unbounded node storage capacities, piecewise constant edge capacities $u_{ij}$ and a constant and uniform edge length $\tau$ on all edges, Let $\mu_{ij}$ be such that $u_{ij}$ has at most $\mu_{ij}$ pieces. Then the maximum $s$-$d$ flow value on $N$ over time horizon $T$ can be found in time $MF(O(n^2\mu),O(mn\mu))$, where $n=|V|,m=|E|,$ $\mu:=\sum_{ij\in E}\mu_{ij}$, and $MF(n',m')$ is the runtime for a steady-state (classical) max flow algorithm on a network with $n'$ nodes and $m'$ edges.
\end{theorem}

\begin{proof}
Lemmas \ref{lem:criticaltimes} and \ref{lem:breakpoints} imply that we just need to solve the steady-state max flow problem on $\cten(N,A,T)$ with $A:=\critnet{N}$.  As $|A|=O(n\mu)$, $\cten(N,A,T)$ has $n\cdot |A|=O(n^2\mu)$ nodes. This leaves us with bounding the number of edges in $\cten(N,A,T)$.

We will count two types of edges. Let $A$ consist of $t_1<t_2<\cdots<t_p$

\begin{itemize}
    \item {\bf Storage edges.} Each node $(i,t_l)$ for $l<p$ has an edge to $(i,t_{l+1})$, which we call a storage edge. The number of storage edges is bounded by the number of nodes, which is $O(n^2\mu)$.

    \item {\bf Transmission edges.} Each node $(i,t_l)$ has an edge to $(j,t_{l'})$ iff there exists $t\in [t_l,t_{l+1}-1]$ such that: (1) $u_{ij}(t)>0$ and (2) $t+\tau\in [t_{l'},t_{l'+1}-1]$.

    To make our upper bound slightly simpler, we will count the number of edges $\cten(N,A,T)$ contains when we instead include an edge from $(i,t_l)$ to $(j,t_l')$ iff:

    \begin{enumerate}
        \item $ij\in E$ (i.e. there exists a time $t$ such that $u_{ij}(t)>0$), and

        \item there exists $t\in [t_l,t_{l+1}-1]$ such that $t+\tau\in [t_{l'},t_{l'+1}-1]$.
    \end{enumerate}
    
Note that this may increase the number of edges in the cTEN (by adding some edges of $0$ capacity), but it never decreases the number of edges.


    Fix $i$ and $j$ such that $ij\in E$ and let $S_l$ be the set of $l'$ such that $(i,t_l)$ has an edge to $(j,t_{l'})$. Then the number of edges from some $(i,t)$ to some $(j,t')$ for all choices of $t,t'\in A$ is $\sum_{l=1}^p|S_l|$.

    Note that $k+1\in S_l$ implies $k\notin S_{l'}$ for all $l'>l$. Further, $k+1\in S_{l-1}$ implies $k'\notin S_l$ for all $k'<k+1$.  These are both true by condition (2) and the fact that $\tau$ is static. Thus, if some index $k$ first appears in set $S_l$ and last appears in set $S_{l'}$, then $|S_{l''}|=1$ for all $l<l''<l'$.

     Let $n_k$ be the number of sets $S_l$ that $k$ appears in. If $n_k\geq 3$, then $n_k-2$ sets $S_l$ (namely, all sets between the first and the last set containing $k$) must contain {\em only} $k$. 
     
     \begin{align*}
    \sum_{l=1}^{p}|S_l| &=\sum_{k=1}^pn_k \\
    &\leq \sum_{k:n_k\leq 2}2+\sum_{k:n_k>2}n_k \\
    &\leq 2p+\sum_{k:n_k>2}n_k.
     \end{align*}

As deduced above, we also know that there must be at least $\sum_{k:n_k>2}(n_k-2)$ singleton sets. The number of singleton sets cannot be larger than the total number of sets, which is $p$, so we get 

\begin{align*}
    p &\geq \sum_{k:n_k>2}(n_k-2) \\
    &=\sum_{k:n_k>2}n_k-\sum_{k:n_k>2}2 \\
    &\geq \sum_{k:n_k>2}n_k-2p,
\end{align*}
 as there are only $p$ possible values of $k$. This gives us $\sum_{k:n_k>2}n_k\leq 3p$. Thus, $\sum_{l=1}^p|S_l|\leq 2p+3p=5p$, where $p=|A|$. When $A=\critnet{N}$, $|A|\leq (2n+1)\mu$. Notably, this argument used the fact that the length of edge $ij$ was static (not changing over time), but it did not use the fact that it was the same as all other edge lengths, and uniformity over edge lengths is actually not necessary to reach this conclusion. 

 There are $m$ possible choices of $i$ and $j$ (as $ij\in E$), so we get that there are $O(nm\mu)$ transmission edges.
    
\end{itemize}

Thus, in total we have $O(n^2\mu)+O(nm\mu)=O(nm\mu)$ edges in $\cten(N,A,T)$ (assuming $m=\Omega( n)$, as otherwise we can remove any nodes not incident to an edge).
\end{proof}




Note that once a maximum flow value is known for the given network, we  can use the quickest transshipment algorithm of Anapolska et al. \cite{anapolska2025} and the temporal-to-static transshipment reduction of Hoppe and Tardos \cite{hoppe2000} to find a specific maximum flow in strongly polynomial time $\tilde O(\mu^7)$.

%% file: proofs.tex
\section{Finding the value of the max flow\label{sec:flow-value}}

In this section we  prove Lemma~\ref{lem:breakpoints}. Our goal is to show the existence of a min cut $\phi$ in $\ten(N,T)$ with $\critcut{\phi}\cap [0,T]\subseteq \critnet{N}$. At a high level, our argument proceeds as follows: We start with an arbitrary min cut. If it does not satisfy the property we need, we show that we can locally modify it without increasing its capacity. 




We first define the local operation on cuts that allows us to move from one min cut to another. Let $\phi$ denote a cut function in $\ten(N,T)$. We use $X_\phi$ to denote vertices that are mapped to $0$ or $T+1$ in $\phi$: $X_\phi:=\{i\in V: \phi(i)\in\{0,T+1\}\}$. Then for any subset $C\subseteq V\setminus X_\phi$, we can define two variants of the cut function $\phi$ -- one where the assignments to vertices in $C$ move ``up'' and the other where the assignments move ``down''. Formally,
\begin{definition}
    Given a cut function $\phi$ in $\ten(N,T)$ and a set $C\subseteq V\setminus X_\phi$, define the cut functions $\phi_C^+$ and $\phi_C^-$ as follows:
    \begin{itemize}
        \item $\phi_C^+(i) := \phi(i)+1$ for all $i\in C$, and $\phi_C^+(i) := \phi(i)$ otherwise. 
        \item $\phi_C^-(i) := \phi(i)-1$ for all $i\in C$, and $\phi_C^-(i) := \phi(i)$ otherwise. 
    \end{itemize}
\end{definition}



Next we identify properties the set $C$ needs to satisfy to ensure that $\phi_C^+$ and $\phi_C^-$ have the same (minimum) capacity as $\phi$. Informally, vertices outside of $C$ define certain ``forbidden'' points of time for vertices in $C$, and for each $i\in C$ we will decide if $\phi_C^+(i)$ and $\phi_C^-(i)$ differ from $\phi(i)$ based on whether or not $\phi(i)$ is in this forbidden set. If $\phi(i)$ is in the forbidden set, we adjust $\phi(i)$ by moving it up one (so it  switches to $s$'s side of the cut later on) or down one (so it switches to $s$'s side of the cut earlier on) in $\phi_C^+$ or $\phi_C^-$, respectively.
Throughout the rest of this section, let $\tau$ be the static travel time of all edges in the network $N$ and let $\cT:=\set{t|\exists ij\in E:u_{ij}(t)\neq u_{ij}(t-1)}\cup\set{0,T,T+1}$.


\begin{definition}
Given a cut function $\phi$ in $\ten(N,T)$, a set $C\subseteq V\setminus X_\phi$, and some vertex $i\in C$, define $\forbidden{\phi}{C}{i}$ as the union of the following sets: 
\begin{itemize}
        \item $\breaks\cup\{\theta+\tau : \theta\in\breaks\}$ 
        \item $\{\phi(j)-\tau,\phi(j) : ij\in E \text{ and } j\in V\setminus C\}$ 
        \item $\{\phi(j)+\tau,\phi(j) : ji\in E \text{ and } j\in V\setminus C\}$ 
\end{itemize}
\end{definition}

We can prove that if $\phi(i)$ does not fall into one of the ``forbidden'' points for $(i,\phi,C)$, then $\phi_C^+$ and $\phi_C^-$ have the same cost as $\phi$. 

\begin{claim}
\label{lem:local-op}
    Given a network $N$ and min $s$-$d$ cut function $\phi$, let $C\subseteq V\setminus X_\phi$ satisfy $\phi(i)\notin \forbidden{\phi}{C}{i}$ for all $i\in C$. Then $cost(\phi_C^+)=cost(\phi_{C}^-)=cost(\phi)$.
\end{claim}

We defer the proof of this claim to the end of this section, but we briefly outline the idea here. Note that some number of copies of each edge $ij\in E$ crosses the cut defined by $\phi$. We want to show that if $\phi_C^+$ incurs some additional cost by having an extra copy of $ij$ cross it as compared to $\phi$, then $\phi_C^-$ obtains an equivalent ``discount'' by having one less copy of $ij$ cross it as compared to $\phi$. The reverse is also true -- that is $\phi_C^+$ gets a total discount value equal to the total extra cost value of $\phi_C^-$. As neither $\phi_C^+$ nor $\phi_C^-$ can be cheaper than $\phi$, it must be the case that all three cuts have the same cost. 

To give some intuition for the offsetting discounts and extra costs, note that the number of ``copies'' of $ij$ that cross the cut $\phi$ is at most $\max\set{0,\Delta_{ij}^\phi}$, where $\Delta_{ij}^\phi:=\phi(j)-(\phi(i)+\tau)$. If $i,j\in C$ or $i,j\in V\setminus C$, $\Delta_{ij}^\phi=\Delta_{ij}^{\phi_C^+}=\Delta_{ij}^{\phi_C^-}$, as $\phi(i)$ and $\phi(j)$ change by the same amount.
Alternatively, consider the case that $i\in C,j\in V\setminus C$. We have $\Delta_{ij}^{\phi_C^-}=\phi(j)-(\phi(i)-1+\tau)=\Delta_{ij}^\phi+1$. Likewise, $\Delta_{ij}^{\phi_C^+}=\Delta_{ij}^\phi-1$. Thus, if $\max\set{0,\Delta_{ij}^{\phi_C^+}}-\max\set{0,\Delta_{ij}^{\phi}}\neq \max\set{0,\Delta_{ij}^{\phi}}-\max\set{0,\Delta_{ij}^{\phi_C^-}} $ (i.e. the extra cost incurred by $\phi_C^+$ for this edge is not equal to the extra discount incurred by $\phi_C^-$), then we have $\max\set{0,\Delta_{ij}^\phi+1}-\max\set{0,\Delta_{ij}^\phi}\neq \max\set{0,\Delta_{ij}^\phi}-\max\set{0,\Delta_{ij}^\phi-1}$, which implies that $0=\Delta_{ij}^\phi=\phi(j)-(\phi(i)+\tau)$ (as $\Delta_{ij}^\phi$ must be an integer). However, since $j\in C$ and $ij\in E$, $\phi(j)-\tau$ is a forbidden point for $\phi(i)$ and this is impossible. Figure \ref{fig:bad-edges} gives a visual sense of this phenomenon, and the proof of Claim \ref{lem:local-op} at the end of the section gives a full and more  careful analysis.




\begin{figure}
    \centering
    \includegraphics[width=.48\textwidth]{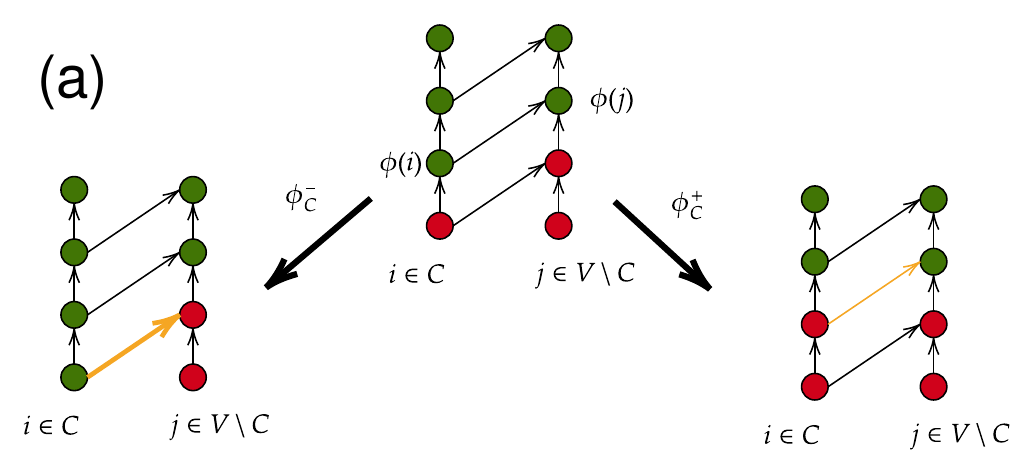}
    \includegraphics[width=.48\textwidth]{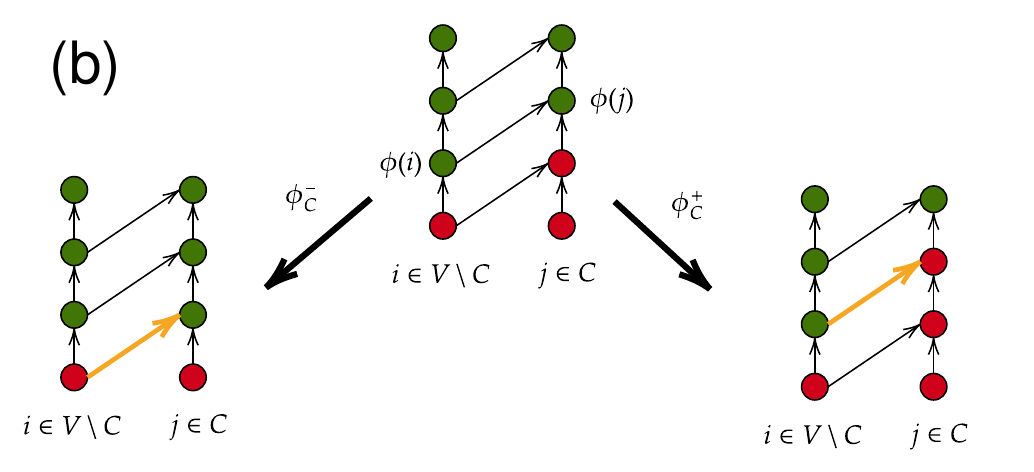}
    \caption{This figure provides a visualization of why certain points are forbidden when trying to shift $\phi$ up or down. In this example, the travel time is $1$, and green nodes are on $(s,0)$'s side of the cut, and red nodes are on $(d,T)$'s side. Thus, the edges that incur a cost for this cut are edges that go from green to red.   \\
    Subfigure (a) depicts a problem with edge cancellation in the proof of Claim \ref{lem:local-op} when there exists $i\in C$ and  $j\in V\setminus C$ with $\phi(i)=\phi(j)+\tau$. In particular for this graph, $\movedown{C}$ incurs an extra cost for edge $(i,\phi(i)-1)(j,\phi(j)-1)$ but $\moveup{C}$ does not incur an equivalent discount for the edge $(i,\phi(i))(j,\phi(j))$, as this edge was already free under $\phi$.
    \\
    Subfigure (b) depicts a problem with edge cancellation in the proof of Claim \ref{lem:local-op} when there exists $i\in V\setminus C$ and $j\in C$ with $\phi(i)=\phi(j)+\tau$. In particular, for this graph $\moveup{C}$ incurs an extra cost for the edge $(i,\phi(i))(j,\phi(j))$ but $\moveup{C}$ does not incur an equivalent discount for the edge $(i,\phi(i)-1)(j,\phi(j)-1)$, as this edge was already free under $\phi$.
    }
    \label{fig:bad-edges}
\end{figure}



Now we define auxillary graphs $G_\phi$ for each min cut $\phi$ that will help us in our proof. 

\begin{definition}
    The {\em pinned assignments graph} $G_\phi$ for $\ten(N,T)$ is an unweighted undirected graph on $V\cup \breaks$ with an edge set defined as follows: 
    \begin{itemize}
        \item Any edge $\set{i,j}$ is in $G_\phi$ for $i,j\in V$ iff $|\phi(i)-\phi(j)|\in \set{0,\tau}$.
        \item An edge $\set{i,\theta}$ is in $G_\phi$ for $i\in V$ and $\theta\in\breaks$ iff $|\phi(i)-\theta|\in \set{0,\tau}$.
    \end{itemize}
\end{definition}

We can now prove Lemma \ref{lem:breakpoints}. 

\begin{proof}[Proof of Lemma \ref{lem:breakpoints}]
Let $cost\st$ be the cost of a minimum $(s,0)$-$(d,T)$ cut in $\ten(N,T)$
For contradiction, suppose the lemma is false. That is, there does not exist a cut $\critcut{\phi}\cap [0,T]\subseteq \critnet{N}$. Among all cut functions of cost $cost\st$, let $\phi$ be a cut function that maximizes $\sum_{i\in V}\phi(i)$.
We show that since $\phi$ does not have $\critcut{\phi}\subseteq \critnet{N}$, we can find a new cut function $\phi'$ of the same cost as $\phi$ such that $\sum_{i\in V}\phi'(i)$ is larger than $\sum_{i\in V}\phi(i)$, thus getting a contradiction and proving our result. 

Recall that $G_\phi$ is the pinned assignments graph for the cut function $\phi$.
Consider a connected component of $G_\phi$ that does not contain any $\theta\in \cT$. Such a component must exist, as otherwise every node $i\in V$ has a path in $G_\phi$ to some $\theta\in \cT$. Let $(i=i_0,i_1,\ldots,i_k,\theta)$ be such a path. Then by definition of $G_\phi$, for each $0<j\leq k$, $|\phi(i_{j-1})-\phi(i_j)|\in \set{0,\tau}$, and $|\phi(i_k)-\theta|\in \set{0,\tau}$. However, this implies that $\phi(i)$ is offset from $\theta$ by $\ell\cdot \tau$, for some $\ell\in [-(k+1),(k+1)]$. However, there are only $n$ nodes in $V$, so $k\leq n-1$ and we get that $\phi(i)$ must be in $\critnet{N}$. Thus, there must be some $i$ without a path to any $\theta\in \cT$, and we select that connected component as $C$. 

We claim that for all $i\in C$, $\phi(i)\notin \forbidden{\phi}{i}{C}$. If this were not the case, then either (1) there exists some $i\in C,j\notin C$ such that $ij\in E$ and $|\phi(j)-\phi(i)|\in \set{0,\tau}$, or (2) there exists some $i\in C$ and $\theta\in \cT$ such that $|\phi(i)-\theta|\in \set{0,\tau}$. In the first case, there must be an edge in $G_\phi$ from $i$ to some $j$ not in $C$, which contradicts the definition of a connected component. In the second case, there must be an edge in $G_\phi$ from some $i\in C$ to some $\theta\in \cT$, which would imply that $\theta$ is in $C$. However, we chose $C$ so that it did not contain any such points. Thus, for all $i\in C$, $\phi(i)\notin \forbidden{\phi}{i}{C}$ and by Claim \ref{lem:local-op}, $cost(\phi_C^+)=cost(\phi)=cost\st$. Recall that for all $i$, $\phi_C^+(i)\geq \phi(i)$, and the inequality is strict when $i\in C$. Since $C$ contains at least one point in $V$,  $\sum_{j\in V}\phi_C^+(j)>\sum_{j\in V}\phi(j)$, but since $\phi_C^+$ defines a min cut, this contradicts how we picked $\phi$.
\end{proof}

%% file: appendix.tex
\subsection{Proof of Claim~\ref{lem:local-op} \label{sec:extra-proof}}


\begin{proof}[Proof of Claim~\ref{lem:local-op}]
Claim \ref{lem:local-op} is true for general networks with static (but not necessarily uniform) edge lengths. Thus we consider a more general version of $\forbidden{i}{C}{\phi}$, where  the set of forbidden points for $(i,\phi,C)$ is the union of the following sets:

\begin{itemize}
        \item $\breaks\cup\{\theta+\tau_{ij} : \theta\in\breaks\}$ 
        \item $\{\phi(j)-\tau_{ij},\phi(j) : ij\in E \text{ and } j\in V\setminus C\}$ 
        \item $\{\phi(j)+\tau_{ji},\phi(j) : ji\in E \text{ and } j\in V\setminus C\}$,
\end{itemize}

and we show that if $\phi(i)\notin \forbidden{i}{C}{\phi}$ for all $i\in C$, ten $cost(\moveup{C})=cost(\movedown{C})=cost(\phi)$. (Note that if $\tau_{ij}=\tau$ for all $ij\in E$, the forbidden set we've defined here is the same as in the original definition we gave.)

Because $\phi$ is a min cut, $cost(\phi)\leq cost(\moveup{C}),cost(\movedown{C})$. Thus, we must have that $(cost(\moveup{C})-cost(\phi))+(cost(\movedown{C})-cost(\phi))\geq 0$. If we can show equality, we get that $cost(\phi)=cost(\moveup{C})=cost(\movedown{C})$ and we are done.

Let $S_\phi:=\set{(i,t):\phi(i)\leq t}$ and let $D_\phi:=\set{(i,t):\phi(i)>t}$.

    \begin{enumerate}
        \item 
        Consider the contribution of an edge $(i,t)(j,t+\tau_{ij})$ to the cost of $\phi$ versus the contribution of the same edge to the cost of $\moveup{C}$. Note that the contribution is identical unless either $i\in C,t=\phi(i)$ or $j\in C,t+\tau_{ij}=\phi(j)$ (i.e. at least one endpoint changed sides of the cut). 

        First consider the case $i\in C,t=\phi(i)$. We have $(i,t)\in S_\phi,(i,t)\in D_{\moveup{C}}$. Thus, if $(j,t+\tau_{ij})\in D_\phi$ (i.e. $\phi(j)>t+\tau_{ij}$), we get a ``discount'' of $u_{ij}(t)$ for this edge, as $(j,t+\tau_{ij})\in D_{\moveup{C}}$ as well. If this is not the case, the edge costs us nothing under either cut.

        Now consider the case that $j\in C,t+\tau_{ij}=\phi(j)$. Note that this implies $(j,t+\tau_{ij})\in S_\phi$, so the edge $(i,t)(j,t+\tau_{ij})$ costs nothing under $\phi$. However, $(j,t+\tau_{ij})\in D_{\moveup{C}}$, so if $(i,t)\in S_{\moveup{C}}$, we incur a new cost $u_{ij}(t)$ for this edge under $\moveup{C}$. 
        If $i\in V\setminus C$, $(i,t)\in S_{\moveup{C}}$ if and only if $t\geq \moveup{C}(i)=\phi(i)$ , so we are in the expensive case if and only if $t\geq \phi(i)$. 
        If $i\in C$, then $(i,t)\in S_{\moveup{C}}$ if and only if $t\geq \moveup{C}(i)=\phi(i)+1$, so we are in the expensive case if and only if $t\geq \phi(i)+1$.



        Thus, putting together the extra costs and extra discounts we obtain from moving these endpoints and letting $u_{ij}(t)$ be $0$ for all times $t$ if the edge from $i$ to $j$ does not exist, we get 

        \begin{align*}
            cost(\moveup{C})-cost(\phi) &= \sum_{i\in C}\sum_{j:\phi(j)>\phi(i)+\tau_{ij}}-u_{ij}(\phi(i)) \\
            & \ \ \ \ \ \ \ \ + \sum_{j\in C}\sum_{i\in C:\phi(j)\geq \phi(i)+\tau_{ij}+1}u_{ij}(\phi(j)-\tau_{ij}) \\
            & \ \ \ \ \ \ \ \ + \sum_{j\in C}\sum_{i\in V\setminus C: \phi(j)\geq \phi(i)+\tau_{ij}}u_{ij}(\phi(j)-\tau_{ij})
        \end{align*}

    \item Now consider $cost(\movedown{C})-cost(\phi)$. We can proceed similarly to the first case. In this case, the contribution of an edge $(i,t)(j,t+\tau_{ij})$ is the same under the two cut functions unless either $i\in C,t=\phi(i)-1$ or $j\in C,t+\tau_{ij}=\phi(j)-1$. (Note that in this case, $(\ell,\phi(\ell)-1)$ is the only node of the \ten\ corresponding to $\ell\in V$ that may change sides of the cut, as an {\em additional} point, namely $(\ell,\phi(\ell)-1)$, is added to $S_{\movedown{C}}$, rather than the point above this being removed from $S_{\moveup{C}}$ as in the previous case.)

    Deal first with the case $i\in C,t=\phi(i)-1$. In this case, $(i,t)\in S_{\movedown{C}}$ and $(i,t)\in D_\phi$. thus, we incur an extra cost for this edge if and only if $(j,t+\tau_{ij})\in D_{\movedown{C}}$. If $j\in C$, this is true if and only if $t+\tau_{ij}<\movedown{C}(j)=\phi(j)-1$. If $j\in V\setminus C$, this is true if and only if $t+\tau_{ij}<\movedown{C}(j)=\phi(j)$. In either of these cases, we incur an additional cost of $u_{ij}(t)$. 

    Now consider the case that $j\in V\setminus C,t+\tau=\phi(j)-1$. 
    In this case, $(j,t+\tau_{ij})\in D_\phi$ and $(j,t+\tau_{ij})\in S_{\movedown{C}}$, so we incur a discount of $u_{ij}(t)$ for this edge under $\movedown{C}$ if $(i,t)\in S_\phi$ (which implies $(i,t)\in S_{\movedown{C}}$). This is true if and only if $t\geq \phi(i)$. 

    Putting together the extra costs and discounts, we get 
    \begin{align*}
        cost(\movedown{C})-cost(\phi) &= \sum_{i\in C}\sum_{j\in C:\phi(j)>\phi(i)+\tau_{ij}} u_{ij}(\phi(i)-1) \\
        &  \ \ \ \ \ \ \ \ + \sum_{i\in C}\sum_{j\in V\setminus C:\phi(j)>\phi(i)+\tau_{ij}+1}u_{ij}(\phi(i)-1) \\ 
        & \ \ \ \ \ \ \ \ + \sum_{j\in C}\sum_{i\in V: \phi(j)\geq \phi(i)+\tau_{ij}+1}-u_{ij}(\phi(j)-\tau_{ij}-1)
    \end{align*}

    \end{enumerate}

    Now we consider $(cost(\movedown{C})-cost(\phi))+(cost(\moveup{C})-cost(\phi))$, and we get

\begin{align*}
 &= \sum_{i\in C}\sum_{j:\phi(j)>\phi(i)+\tau_{ij}}-u_{ij}(\phi(i)) + \sum_{j\in C}\sum_{i\in C:\phi(j)\geq \phi(i)+\tau_{ij}+1}u_{ij}(\phi(j)-\tau_{ij}) \\
& \ \ \ \ \ \ \ \ + \sum_{j\in C}\sum_{i\in V\setminus C: \phi(j)\geq \phi(i)+\tau_{ij}}u_{ij}(\phi(j)-\tau_{ij})  + 
\sum_{i\in C}\sum_{j\in C:\phi(j)>\phi(i)+\tau_{ij}} u_{ij}(\phi(i)-1) \\
&  \ \ \ \ \ \ \ \ + \sum_{i\in C}\sum_{j\in V\setminus C:\phi(j)>\phi(i)+\tau_{ij}-1}u_{ij}(\phi(i)-1)  \\
& \ \ \ \ \ \ \ \  + \sum_{j\in C}\sum_{i: \phi(j)\geq \phi(i)+\tau_{ij}-1}-u_{ij}(\phi(j)-\tau_{ij}-1) \\
&= \sum_{i\in C}\bigg( \sum_{j:\phi(j)>\phi(i)+\tau_{ij}}-u_{ij}(\phi(i)) + \sum_{j\in C:\phi(j)>\phi(i)+\tau_{ij}}u_{ij}(\phi(i)-1) \\
& \ \ \ \ \ \ \ \  + \sum_{j\in V\setminus C:\phi(j)>\phi(i)+\tau_{ij}-1}u_{ij}(\phi(i)-1) \Bigg) \\
& \ \ \ + \sum_{j\in C}\Bigg( \sum_{i\in C:\phi(j)\geq \phi(i)+\tau_{ij}+1}u_{ij}(\phi(j)-\tau_{ij}) + \sum_{i\in V\setminus C:\phi(ij)\geq \phi(i)+\tau_{ij}}u_{ij}(\phi(j)-\tau_{ij}) \\
& \ \ \ \ \ \ \ \ + \sum_{i:\phi(j)\geq \phi(i)+\tau_{ij}+1}-u_{ij}(\phi(j)-\tau_{ij}-1)\Bigg)
\end{align*}

    

    First, fix an $i$ and consider a term of the first sum corresponding to $i$. In the case that $u_{ij}(\phi(i)-1)$ appears as a positive term in this sum, either ($\phi(j)>\phi(i)+\tau_{ij}$ and $j\in C$) or ($\phi(j)>\phi(i)+\tau_{ij}-1$ and $j\in V\setminus C$). In the first  case, $-u_{ij}(\phi(i)-1)$ also appears in the sum. In the second case, $-u_{ij}(\phi(i)-1)$ appears in the sum unless $\phi(j)=\phi(i)-\tau_{ij}$, which is impossible since $j\in V\setminus C$ and $\phi(j)=\phi(i)-\tau_{ij}$ implies $\phi(i)\in \forbidden{i}{C}{\phi}$. Thus, if $u_{ij}(\phi(i))$ appears in the sum, so does $-u_{ij}(\phi(i)-1)$ also appears in the sum. Further, $u_{ij}(\phi(i)-1)=u_{ij}(\phi(i))$, as otherwise $\phi(i)\in \cT$, which is impossible since $i\in C$ and the conditions of the claim thus forbid that.

    Now conisder the second sum and fix some $j\in C$. If $u_{ij}(\phi(j)-\tau_{ij})$ appears in the sum, then either ($i\in C$ and $\phi(j)\geq \phi(i)+\tau_{ij}+1$) or ($i\in V\setminus C$ and $\phi(j)\geq \phi(i)+\tau$). In the first case, $-u_{ij}(\phi(j)-\tau_{ij}-1)$ also appears in the sum. In the second case, $-u_{ij}(\phi(j)-\tau_{ij}-1)$ appears in the sum unless $\phi(j)=\phi(i)+\tau_{ij}$, which is impossible since $j\in C$ and $i\notin C$, so $\phi(i)+\tau_{ij}\in \forbidden{j}{C}{\phi}$. Thus, if $u_{ij}(\phi(j)-\tau_{ij})$ appears in the sum, $-u_{ij}(\phi(j)-\tau_{ij}-1)$ also appears in the sum. Further, $u_{ij}(\phi(j)-\tau_{ij})=u_{ij}(\phi(j)-\tau_{ij}-1)$, as otherwise $\phi(j)=\theta+\tau_{ij}$ for some $\theta\in \cT$,  which is impossible since $j\in C$ and the conditions of the claim forbid that.
    

\end{proof}

%% file: other_times.tex
\section{A note on networks with a small number of distinct edge lengths} \label{sec:otherlengths}

The algorithmic approach we have discussed so far relies on the assumption that all edges in the network have a uniform and static edge length $\tau$. In fact, this approach can be generalized to other types of networks with static (but not uniform) edge lengths, but the final runtime depends exponentially on the number of distinct edge lengths.

In particular, the only place we used the uniformity of edge lengths was in the proof of Lemma \ref{lem:breakpoints}, where we considered a pinned assignments graph that placed an edge between $i$ and $j$ if and only if $|\phi(i)-\phi(j)|\in \set{0,\tau}$. This allowed us to argue that if a node $i\in V$ in this graph had some path back to a breakpoint $\theta\in \cT$, $\phi(i)$ must be offset from $\theta$ by a multiple of $\tau$. Suppose instead we let $\Gamma$ be the set of possible edge lengths in our network, and we added $\set{i,j}$ to $G_\phi$ if and only if $|\phi(i)-\phi(j)|\in \set{0}\cup \Gamma$. Likewise, we can add $\set{i,\theta}$ to $G_\phi$ if and only if $|\phi(i)-\theta|\in \set{0}\cup \Gamma$. Then if $i$ has a path to $\theta\in \cT$ through $G_\phi$, a shortest such path must take the form $(i=i_0,i_1,\ldots,i_k,\theta)$ and there must exist coefficients $a_\gamma\in [-(k+1),(k+1)]$ for each $\gamma\in \Gamma$ such that $|\phi(i)-\theta|=\sum_{\gamma\in \Gamma}a_{\gamma}\cdot \gamma$. (In fact, $\sum_{\gamma\in \Gamma}|a_\gamma|\leq k+1$, but we use a simpler bound by bounding the range of each $a_\gamma$ individually.)

Thus, if we let the critical points for this network be all points $\theta +\sum_{\gamma \in \Gamma}a_{\gamma}\cdot \gamma$ for $a_{\gamma}\in [-n,n]$, we get that when the range of $\phi $ is not a subset of the critical points, there must exist $i\in V$ that is not connected to anything in $\cT$ in $G_\phi$ and the proof of Lemma \ref{lem:breakpoints} proceeds as before. 

There are at most $(2n)^{|\Gamma|}$ ways to select the coefficients $a_\gamma$, so the critical point set has size at most $O(\mu \cdot (2n)^{|\Gamma|})$ and thus the graph we need to run steady-state max flow on has $O(\mu\cdot (2n)^{|\Gamma|+1})$ nodes and $O(\gamma \cdot (2n)^{|\Gamma|+1}\cdot m)$ edges. When $|\Gamma|$ is some small constant $c$, this is $O(\mu \cdot n^{c+1})$ and $O(\mu\cdot n^{c}m)$. Thus, when the network has at most a constant $c$ number of distinct static edge lengths, the problem can be solved in strongly polynomial time $O(\mu^2\cdot m \cdot n^{2c+1})$.